\documentclass[11pt]{article}
 
\usepackage[margin=1in]{geometry}
\usepackage{microtype}
\usepackage{lmodern}
 
\usepackage{amsmath,amssymb,amsthm,mathtools}
\usepackage{dsfont}
 
\usepackage{graphicx}
\usepackage{booktabs}
\usepackage{pgfplots}
 
\usepackage{xcolor}
\usepgfplotslibrary{groupplots}
\pgfplotsset{compat=1.18}
 
\usepackage[textsize=tiny]{todonotes}
 
\usepackage[sort&compress]{natbib}
\usepackage[colorlinks=true,linkcolor=blue,citecolor=blue,urlcolor=blue]{hyperref}
\usepackage[capitalize,noabbrev]{cleveref}
 
\usepackage{enumitem}

\newcommand{\mP}{\mathbb{P}}
\newcommand{\mE}{\mathbb{E}}

\newcommand{\1}{\mathds{1}}
\newcommand{\iid}{\stackrel{\mathrm{i.i.d.}}{\sim}}

\newcommand{\calX}{\mathcal{X}}
\newcommand{\calG}{\mathcal{G}}
 
\newtheorem{proposition}{Proposition}

\newtheorem{theorem}{Theorem}

\newtheorem{corollary}{Corollary}

\newtheorem{assumption}{Assumption}
\newtheorem{remark}{Remark}

\title{More Permutations Do Not Always Increase Power:\\
Non-monotonicity in Monte Carlo Permutation Tests}

\author{
Suman Cha\textsuperscript{1,\dag}
\quad
Seongchan Lee\textsuperscript{2,\dag}
\quad
Antonin Schrab\textsuperscript{3}
\quad
Ilmun Kim\textsuperscript{2}
\\[0.6em]
\small \textsuperscript{1}Department of Statistics and Data Science, Yonsei University, Seoul, South Korea
\\
\small \textsuperscript{2}Department of Mathematical Sciences, KAIST, Daejeon, South Korea
\\
\small \textsuperscript{3}Department of Computer Science and Technology, University of Cambridge, Cambridge, UK
\\[0.4em]
\small \textsuperscript{\dag}These authors contributed equally to this work.
}

\date{\today}

\begin{document}

\maketitle

\begin{abstract}
Monte Carlo permutation tests are a cornerstone of valid, model-free statistical inference. A widely held practical intuition is that increasing the number of sampled permutations improves test performance, in particular that statistical power tends to increase with the Monte Carlo budget. In this paper, we show that these intuitions are false in general. Leveraging the saw-toothed structure of power arising from distributional discreteness, we provide a simple structural explanation for why power can decrease as the number of sampled permutations increases, and we prove that such decreases occur infinitely often as the Monte Carlo budget grows.
\end{abstract}


\section{Introduction}
 
Permutation tests, originated with \citet{fisher1935design} and \citet{pitman1937significance}, provide exact or finite-sample-valid inference under exchangeability (or invariance) assumptions. Owing to this finite-sample validity, they are widely used across modern statistics, machine learning, and the sciences \citep{arboretti2025review, ritzwoller2024randomization}. In practice, since the full permutation group is often too large to enumerate, practitioners typically employ a Monte Carlo approach by drawing $B$ random permutations \citep{dwass1957modified}. The prevailing heuristic suggests that increasing $B$ yields a ``better'' test, as the Monte Carlo $p$-value concentrates around its full-group counterpart with an error shrinking on the order of $B^{-1/2}$ \citep{jockel1986finite}.

Despite this intuition, the relationship between the Monte Carlo budget $B$ and test power remains only partially understood. Historically, the literature has predominantly focused on exact size control, practical stability \citep{marriott1979barnard}, or conditions that guarantee monotonic power increases under the strict alignment condition $(B+1)\alpha \in \mathbb{N}$. Under this aligned setting, monotonicity has been established assuming monotone likelihood ratios \citep{hope1968simplified} or concavity of the underlying power function \citep{jockel1986finite}. More recently, \citet{gaigall2026number} analyzed a broader resampling framework covering bootstrap and permutation tests, bounding rejection probabilities via properties of the $p$-value distribution such as convexity and the Bernstein property. Taken together, these studies clarify power behavior under specific structural constraints, but they do not characterize the general behavior of power for arbitrary choices of $B$. Crucially, the possibility of non-monotonicity, meaning that increasing the number of sampled permutations could actually reduce the rejection probability under the alternative, has not been systematically investigated.

Related phenomena have been observed in other simulation-based and discrete testing frameworks. \citet{davidson2000bootstrap} analyzed the power loss caused by using a finite number of bootstrap replications, and \citet{chernick2002saw} documented a ``sawtooth'' pattern in the power function of exact proportion tests due to the discreteness of the sampling distribution. While this analysis was confined to the binomial setting without a general structural explanation, it hints at the broader impact of discreteness. Closer to our context, \citet{rindt2021consistency} demonstrated that within permutation testing, the \emph{conditional} rejection probability for a fixed dataset $\calX$ need not increase with the Monte Carlo budget $B$. Our work extends these findings in two directions. First, we demonstrate that the \emph{unconditional} power, averaged over $\calX$ under the alternative, also exhibits non-monotone behavior as a function of $B$. Second, we identify a discrete critical-count structure that explains when and why local downward steps arise along plateaus of the rejection threshold.

A separate line of work has generalized permutation tests along an algebraic axis. \citet{ramdas2023permutation} established finite-sample validity for arbitrary subsets of the permutation group and non-uniform sampling distributions, and \citet{koning2024more} showed that a carefully chosen subgroup can outperform the full group in multiple testing. The power gains in that literature arise from the algebraic structure of the chosen permutations, and are therefore orthogonal to the effect of the Monte Carlo budget $B$ studied here.

Despite the theoretical emphasis on aligned choices with $(B+1)\alpha\in\mathbb{N}$, many applied implementations instead use round Monte Carlo budgets (e.g., $B=10^2,10^3,10^4$) without calibrating $B$ to the nominal level $\alpha$. The goal of this article is to show that non-monotonicity is not a pathological curiosity but a fundamental structural feature of Monte Carlo permutation tests, and to argue that choosing $B$ so that $(B+1)\alpha\in\mathbb{N}$ provides a simple design rule that places the test at a local power maximum under the non-degeneracy condition.

The remainder of the paper is organized as follows. In \Cref{sec: setup}, we review the setup of Monte Carlo permutation tests. In \Cref{sec: nonmonotonicity}, we present our main theoretical results on non-monotonicity and infinitely many local maxima in the power curve. In \Cref{sec: concrete example}, we illustrate these phenomena in a simple closed-form example. We conclude with a discussion in \Cref{sec: discussion}, and all code to reproduce the figures and numerical results is available at \url{https://github.com/suman-cha/mc-perm-power}.

\section{Monte Carlo permutation tests} \label{sec: setup}
 
We briefly review exact permutation tests and their Monte Carlo approximations, fix notation, and introduce the key objects that drive our subsequent analysis of power non-monotonicity. We refer to \citet{lehmann2005testing} and \citet{hemerik2018exact} for comprehensive treatments of permutation testing.

\subsection{Exact permutation tests}
 
Let $\calX$ be observed data (e.g., concatenated samples in a two-sample problem), let $T(\calX)\in\mathbb{R}$ be a test statistic, with larger values indicating stronger evidence against the null hypothesis, and let $\calG$ be a finite group of transformations on the data. In permutation tests, $\calG$ is typically the full permutation group on the sample indices or a subgroup thereof induced by the null hypothesis.
 
The exact permutation $p$-value is defined as
\begin{equation}
\label{eq:p_exact}
p_{\mathrm{exact}}(\calX)
=
\frac{1}{|\calG|}\sum_{\pi\in\calG}\1\{T(\calX^\pi)\ge T(\calX)\}.
\end{equation}
Under standard invariance or exchangeability assumptions, the exact
permutation $p$-value is super-uniform under the null hypothesis:
$\mP(p_{\mathrm{exact}}(\calX) \leq \alpha) \leq \alpha$ for all
$\alpha \in [0,1]$; see \citet{hoeffding1952power}.

\subsection{Monte Carlo approximation}
 
When the permutation group $\calG$ is too large to enumerate exhaustively, one approximates the exact $p$-value by random sampling from $\calG$. Specifically, let $\pi_1,\ldots,\pi_B$ be i.i.d.~random transformations drawn uniformly from $\calG$ (with replacement). The Monte Carlo permutation $p$-value is defined as
\begin{equation}
\label{eq:p_mc}
p_B(\calX)
=
\frac{1+\sum_{i=1}^B \1\{T(\calX^{\pi_i})\ge T(\calX)\}}{B+1},
\qquad
\text{reject if } p_B(\calX)\le \alpha.
\end{equation}
The additive constant ``$+1$'' treats the observed statistic as one of the $B+1$ exchangeable draws (equivalently, it corresponds to augmenting the resampled transformations with the identity), which ensures that $p_B(\calX)$ is never zero and that the resulting test is finite-sample valid under the null \citep{phipson2010permutation,hemerik2018exact}.

To analyze the behavior of this test, it is convenient to separate the randomness coming from the data $\calX$ and the randomness coming from the Monte Carlo sampling. Conditional on the observed data $\calX$, define the single-draw exceedance probability
\begin{align}
\label{eq:q_def}
q(\calX)
\;\coloneqq\;
\mP(T(\calX^\pi)\ge T(\calX)\mid\calX)\in[0,1],
\end{align}
where $\pi$ denotes an independent random transformation drawn uniformly from $\calG$. Thus $q(\calX)$ is the conditional probability that a single randomly sampled permutation produces a test statistic at least as extreme as the observed one.
 
Given $\calX$, the Monte Carlo exceedance indicators $\1\{T(\calX^{\pi_i})\geq T(\calX)\}$ are i.i.d.\ Bernoulli random variables with success probability $q(\calX)$. Consequently, the total number of exceedances
\begin{align*}
R_B(\calX)\;\coloneqq\;\sum_{i=1}^B \1\{T(\calX^{\pi_i})\geq T(\calX)\}
\end{align*}
satisfies the conditional binomial law
\begin{equation*}
R_B(\calX)\mid \calX \sim \mathrm{Binomial}\bigl(B,q(\calX)\bigr).
\end{equation*}
This conditional binomial structure is the key property underlying all subsequent power calculations. Finally, define the integer-valued critical count
\begin{equation*}
k_B \;\coloneqq\; \lfloor (B+1)\alpha\rfloor - 1.
\end{equation*}
With this notation, the Monte Carlo rejection event can be written equivalently as
\begin{equation}
\label{eq:rejection_event}
\{p_B(\calX)\leq \alpha\} 
= 
\{R_B(\calX)\leq k_B\}.
\end{equation}
The non-smooth, integer-valued dependence of $k_B$ on $B$ is the fundamental source of the non-monotonic power behavior analyzed in the remainder of the paper.
 
\paragraph{Conditional and unconditional power.}
For a fixed dataset $\calX$, define the conditional rejection probability
\begin{align*}
\phi_B(\calX) \;\coloneqq\; \mP(p_B(\calX)\le\alpha\mid\calX) = \mP(R_B(\calX)\le k_B\mid\calX),
\end{align*}
and, under an alternative distribution for $\calX$, the unconditional power
\begin{align*}
\mathrm{Pow}(B) \;\coloneqq\; \mE[\phi_B(\calX)].
\end{align*}
Our focus is on how $\mathrm{Pow}(B)$ behaves as a function of the Monte Carlo budget $B$.

\section{Power non-monotonicity and sawtooth behavior in $B$} \label{sec: nonmonotonicity}
 
A central mechanism behind the non-monotonicity studied in this section is the integer-valued rejection threshold $k_B = \lfloor (B+1)\alpha\rfloor - 1$. As $B$ increases, $k_B$ can change only at discrete indices and is otherwise constant over intervals of $B$. Whenever $k_B$ remains unchanged as $B$ increases by one, the rejection rule compares an additional Bernoulli trial to the same critical count. In other words, the exceedance count $R_B$ becomes stochastically larger while the cutoff is held fixed. This mechanism makes rejection \emph{strictly harder}, not easier, for datasets with $q(\calX)\in(0,1)$, and it is a key structural source of the downward steps that occur within each plateau of $k_B$, contributing to the sawtooth behavior of the power curve analyzed below.

\paragraph{Plateaus and jumps of $k_B$.}
The sequence $(k_B)_{B\ge 1}$ is piecewise constant and increases by at most one at each step. More precisely, for $B \geq 2$,

\begin{align*}
k_B - k_{B-1} =
\begin{cases}
1, & \text{if } \lfloor (B+1)\alpha\rfloor = \lfloor B\alpha\rfloor + 1,\\
0, & \text{otherwise}.
\end{cases}
\end{align*}
We call the case $k_B = k_{B-1}$ a plateau and the case $k_B = k_{B-1}+1$ a jump. The evolution of $k_B$ thus consists of flat stretches (plateaus) interspersed with upward increments (jumps), and this discrete structure induces the characteristic sawtooth behavior of the power curve.

Before formalizing this mechanism, we impose a mild non-degeneracy condition on the conditional exceedance probability.
 
\begin{assumption}[Non-degeneracy of the conditional exceedance probability]
\label{assumption: strict prob}
The conditional exceedance probability is not almost surely equal to one, i.e.,
\begin{align*}
\mP\bigl(q(\calX)<1\bigr)>0.
\end{align*}
\end{assumption}
The most interesting regime for power analysis is the non-degenerate case in which the conditional exceedance probability $q(\calX)$ can take values strictly below $1$ with positive probability. Note that the lower bound $q(\calX)>0$ holds automatically. Indeed, the permutation group $\calG$ always contains the identity transformation $\mathrm{id}$, for which $\calX^{\mathrm{id}}=\calX$ and hence $T(\calX^{\mathrm{id}})=T(\calX)$. Consequently,
\begin{align*}
q(\calX)
=
\mP\big(T(\calX^\pi)\ge T(\calX)\mid\calX\big)
\ge
\frac{1}{|\calG|}
>0.    
\end{align*}
Thus, the substantive restriction in \Cref{assumption: strict prob} is that $q(\calX)$ is not almost surely equal to $1$.

If $q(\calX)$ is almost surely equal to $1$, then for every dataset all permutations would produce statistics at least as large as the observed one, implying $p_B(\calX)=1$ for all $B$. The Monte Carlo test would therefore never reject, and the unconditional power would be identically zero and hence trivially monotone in $B$. By contrast, when $q(\calX)<1$ with positive probability, the permutation distribution sometimes contains values strictly smaller than the observed statistic. In this non-degenerate regime, the number of exceedances among the $B$ sampled permutations varies with the Monte Carlo draws, allowing the discrete plateau--jump structure of $k_B$ to translate into genuine non-monotonicity of the unconditional power.

\begin{theorem}[Non-monotonicity and infinitely many local maxima]
\label{thm:nonmonotonicity}
Fix $\alpha\in(0,1)$ and consider the Monte Carlo permutation test~\eqref{eq:p_mc}, with $q(\calX)$, $k_B$, and $\mathrm{Pow}(B)$ as defined in \Cref{sec: setup}. Under \Cref{assumption: strict prob} and $B \geq 2$, the following hold.

\begin{enumerate}[leftmargin=2em]
  \item[(i)] If $B$ satisfies
  \begin{align*}
    k_B-k_{B-1}=1
    \qquad\text{and}\qquad
    k_{B+1}-k_B=0,      
  \end{align*}
  then $\mathrm{Pow}(B)$ is a strict local maximum, i.e.,
  \begin{align*}
    \mathrm{Pow}(B-1)<\mathrm{Pow}(B)
    \quad\text{and}\quad
    \mathrm{Pow}(B+1)<\mathrm{Pow}(B).      
  \end{align*}
  \item[(ii)] There exist infinitely many integers $B$ satisfying the above jump--plateau pattern (that is, a jump at $B$ followed by a plateau at $B+1$, so that $k_{B-1}<k_B=k_{B+1}$). Consequently, the power curve $B\mapsto \mathrm{Pow}(B)$ has infinitely many strict local maxima.
\end{enumerate}
\end{theorem}

\begin{remark}\normalfont
\label{rem:half-alpha}
When $\alpha\in(0,1/2]$, the two conditions in \Cref{thm:nonmonotonicity}(i) are not independent: the jump condition $k_B-k_{B-1}=1$ alone implies $k_{B+1}-k_B=0$. Indeed, letting $\mathrm{frac}\{x\}\coloneqq x-\lfloor x\rfloor$ denote the fractional part of $x$, a jump at $B$ means $\mathrm{frac}\{B\alpha\}\geq 1-\alpha$, which forces $\mathrm{frac}\{(B+1)\alpha\} = \mathrm{frac}\{B\alpha\}+\alpha-1 \in [0,\alpha)$. A jump at $B+1$ would additionally require $\mathrm{frac}\{(B+1)\alpha\}\geq 1-\alpha$. But for $\alpha \leq 1/2$ we have $[0,\alpha)\subseteq[0,1-\alpha)$, so $\mathrm{frac}\{(B+1)\alpha\}<1-\alpha$, ruling out a jump at $B+1$. Hence, for $\alpha\in(0,1/2]$, every jump index $B$ is automatically a local maximizer of the power curve.
\end{remark}

\paragraph{A simple practical criterion.}
A convenient sufficient condition for the jump--plateau pattern in \Cref{thm:nonmonotonicity}(i) is $(B+1)\alpha\in\mathbb{N}$. Indeed, if $(B+1)\alpha$ is an integer, then $B\alpha$ cannot be an integer (otherwise $\alpha=(B+1)\alpha-B\alpha$ would be an integer, contradicting $\alpha\in(0,1)$), so necessarily $k_B-k_{B-1}=1$ and $k_{B+1}-k_B=0$.

Hence, under \Cref{assumption: strict prob}, every such $B$ is a strict local maximizer of the power curve. As we show below (\Cref{cor:rational_levels}), for the common choices $\alpha\in\{0.1,0.05,0.01\}$ this alignment condition also completely characterizes the set of strict local maxima.
 
The next proposition quantifies the size of each downward step and shows that it is $O(B^{-1/2})$ as $B\to\infty$.

\begin{proposition}
\label{prop:decrease_size_uncond}
In the setting of \Cref{thm:nonmonotonicity}, suppose that $k_{B+1}=k_B=:k$. Then
\begin{align*}
\mathrm{Pow}(B)-\mathrm{Pow}(B+1)
=
\mE\bigl[
  q(\calX)\,
  \mP(\mathrm{Binomial}(B,q(\calX))=k)
\bigr],
\end{align*}
where the inner probability is taken with respect to a $\mathrm{Binomial}(B,q(\calX))$ distribution conditional on $\calX$. Moreover, the decrease admits the distribution-free bound
\begin{align*}
\mathrm{Pow}(B)-\mathrm{Pow}(B+1)
\leq
\frac{1}{\sqrt{2\pi (B+1)}}\sqrt{\frac{\alpha}{1-\alpha}}.
\end{align*}
\end{proposition}

\paragraph{A refined design criterion.}
\Cref{thm:nonmonotonicity} guarantees that every jump--plateau index is a strict local maximum of the power curve. The following result provides an explicit characterization of the complete set of such indices. For $\alpha\in(0,1)$, let
\begin{equation}
\label{eq:Balpha_frac}
\mathcal{B}_\alpha 
=
\bigl\{
  B\in\mathbb{N}: \mathrm{frac}\{(B+1)\alpha\}<\min(\alpha,1-\alpha)
\bigr\}.
\end{equation}
Indeed, a jump at $B$ is equivalent to $\mathrm{frac}\{B\alpha\}\geq 1-\alpha$, which forces $\mathrm{frac}\{(B+1)\alpha\}=\mathrm{frac}\{B\alpha\}+\alpha-1\in[0,\alpha)$. A subsequent plateau at $B+1$ is equivalent to $\mathrm{frac}\{(B+1)\alpha\}<1-\alpha$. Combining these conditions yields \eqref{eq:Balpha_frac}.

For $\alpha\in(0,1/2]$, \Cref{rem:half-alpha} shows the plateau condition is automatically satisfied, so \eqref{eq:Balpha_frac} simplifies to
\begin{equation}
\label{eq:Balpha_half}
\mathcal{B}_\alpha
=
\bigl\{
  B\in\mathbb{N}: \mathrm{frac}\{(B+1)\alpha\}<\alpha
\bigr\}
=
\bigl\{
  \lceil j/\alpha\rceil - 1 : j\in\mathbb{N},\, j\geq 1
\bigr\},
\end{equation}
where the second equality holds since $\mathrm{frac}\{(B+1)\alpha\}<\alpha$ if and only if $B+1=\lceil j/\alpha\rceil$ for $j\coloneqq\lfloor (B+1)\alpha\rfloor\geq 1$.
 
\begin{corollary}[Characterization for rational levels]
\label{cor:rational_levels}
When $\alpha=1/m$ for an integer $m\geq 2$, equation~\eqref{eq:Balpha_half} reduces to
\begin{equation}
\label{eq:Balpha_rational}
\mathcal{B}_\alpha
=
\bigl\{
  B\in\mathbb{N}: (B+1)\alpha\in\mathbb{N}
\bigr\}.
\end{equation}
\end{corollary}

The common levels $\alpha\in\{0.1,0.05,0.01\}$ are all of the form $\alpha=1/m$ and thus fall under \Cref{cor:rational_levels}, which follows immediately from \eqref{eq:Balpha_half}: when $\alpha=1/m$, the ceiling $\lceil j/\alpha\rceil=jm$ is always an integer, giving $B=jm-1$, i.e., $(B+1)\alpha=j\in\mathbb{N}$. Together with the direct argument above, this shows that for these levels the alignment condition $(B+1)\alpha\in\mathbb{N}$ not only provides a sufficient condition for a strict local maximum, but also exactly characterizes the entire set $\mathcal{B}_\alpha$.

The results above establish that non-monotonicity of power is a structural feature of Monte Carlo permutation tests and provide sharp finite-$B$ bounds on the size of individual downward steps. We now turn to concrete examples that make this mechanism fully explicit.

\section{A concrete example} \label{sec: concrete example}
We now illustrate the non-monotonicity mechanism in a fully explicit example at level $\alpha=0.05$. This example serves two purposes: first, it yields a closed-form expression for the unconditional power; second, it makes transparent how the integer-valued rejection threshold induces a deterministic sawtooth dependence of the power on the Monte Carlo budget $B$.

Fix an integer $n\geq 2$ and consider a two-group experiment with $n$ units per group and binary outcomes. Under the alternative, the treated group has success probability $p_1\in(0,1)$ and the control group has success probability $p_0=0$.

Let the first $n$ indices correspond to the treated group and the remaining $n$ to the control group. Specifically, define
\begin{align*}
X_1,\ldots,X_n \iid \mathrm{Bernoulli}(p_1), 
\qquad 
Y_1,\ldots,Y_n \iid \mathrm{Bernoulli}(0),    
\end{align*}
and form the pooled outcome vector
\begin{align*}
Z=(Z_1,\ldots,Z_{2n})
\coloneqq
(X_1,\ldots,X_n,Y_1,\ldots,Y_n)
\in\{0,1\}^{2n}.    
\end{align*}
Let
\begin{align*}
S \coloneqq \sum_{i=1}^n Z_i = \sum_{i=1}^n X_i \sim \mathrm{Binomial}(n,p_1)    
\end{align*}
denote the observed number of treated successes. 
Since $Z_{n+1}=\cdots=Z_{2n}=0$ almost surely, the total number of successes in the pooled sample equals $S$. Let $\mathcal{I} \coloneqq \{ I\subset[2n]: |I|=n\}$ denote the collection of all possible treatment labelings of size $n$, so that $|\mathcal I|=\binom{2n}{n}$. The original treatment group is indexed by $I_0 = \{1,\ldots,n\}.$ For any $I\in\mathcal I$, define the treated success count
\begin{align*}
T_I(Z) \coloneqq \sum_{i\in I} Z_i,    
\end{align*}
that is, the number of successes assigned to the treated group under labeling $I$. In particular, under the observed labeling we have $T_{I_0}(Z)=S$.
 
\begin{figure}[t]
\centering
\includegraphics[width=0.9\textwidth]{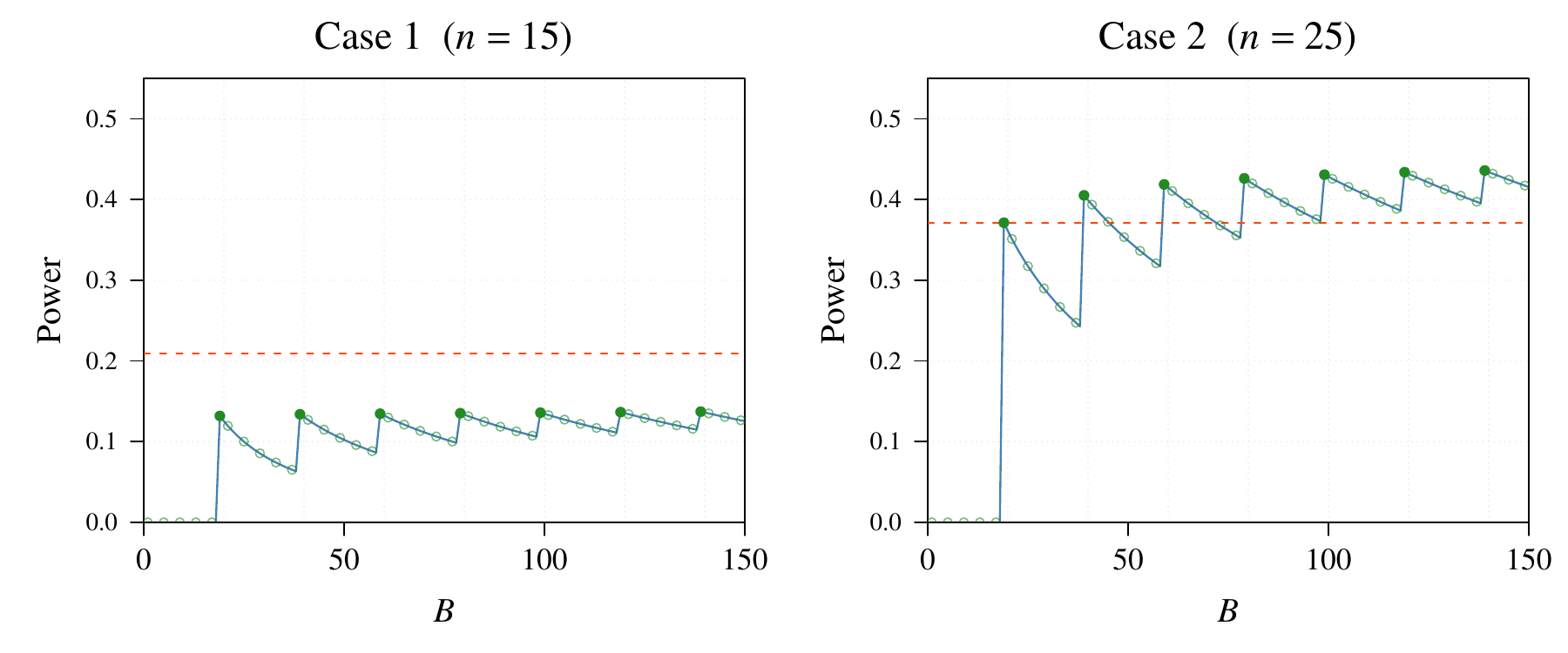}
\caption{Unconditional power curves for the closed-form Bernoulli example at level $\alpha=0.05$ with $p_0=0$ and $p_1=0.16$. Solid line: Monte Carlo power $\mathrm{Pow}(B)$. Dashed line: exact-test power $\mathrm{Pow}_{\mathrm{exact}}$. Dots mark the predicted local maxima where $(B+1)\alpha\in\mathbb N$.}
\label{fig:bernoulli_two_cases}
\end{figure}

\paragraph{Conditional exceedance probability.}
Given $S=s$, the conditional exceedance probability is
\begin{align*}
q(s) \;\coloneqq\; \mP(T_\pi(Z)\ge T_{I_0}(Z)\mid S=s) = \mP(T_\pi(Z)=s\mid S=s),
\end{align*}
where $\pi$ is uniform on $\mathcal I$. The event $\{T_\pi(Z)=s\}$ occurs if and only if the random labeling places all $s$ success indices in the treated group. There are $\binom{2n-s}{n-s}$ such labelings, hence
\begin{align*}
q(s)=
\frac{\binom{2n-s}{n-s}}{\binom{2n}{n}}
\in (0,1] \qquad \text{for $s = 0, 1,\ldots, n$}.
\end{align*}
 
\paragraph{Closed-form power expression.}
With Monte Carlo sampling with replacement,
\begin{align*}
R_B \mid S=s \sim \mathrm{Binomial}\bigl(B,q(s)\bigr), \qquad k_B \coloneqq \lfloor (B+1)\alpha\rfloor - 1,
\end{align*}
and the conditional rejection probability is
\begin{align*}
\phi_B(s) = \mP(R_B \le k_B \mid S=s) = F_{B,q(s)}(k_B),
\end{align*}
where $F_{B,q(s)}$ denotes the cumulative distribution function of $\mathrm{Binomial}(B,q(s))$. Averaging over the data-generating distribution of $S$ yields the unconditional power
\begin{equation}
\label{eq:bernoulli_unconditional_power}
\mathrm{Pow}(B)
=
\mE\big[\phi_B(S)\big]
=
\sum_{s=0}^n
\binom{n}{s}p_1^s(1-p_1)^{n-s}
\,F_{B,q(s)}(k_B).
\end{equation}
This expression is fully explicit and closed-form.
 
\paragraph{Non-monotonicity.}
Here $q(0)=1$, while $q(s)\in(0,1)$ for every $s\in\{1,\dots,n\}$. Thus the non-degeneracy condition \Cref{assumption: strict prob} holds as soon as $\mP(S\ge 1)>0$ (i.e., $p_1\in(0,1)$). In that case, \Cref{thm:nonmonotonicity} implies that the unconditional power in~\eqref{eq:bernoulli_unconditional_power} need not be monotone non-decreasing in $B$.

In this closed-form example, the dependence of $\mathrm{Pow}(B)$ on $B$ is entirely driven by the integer-valued cutoff $k_B$, and the power curve therefore exhibits the deterministic sawtooth pattern described earlier, with strict local maxima at those $B$ for which $(B+1)\alpha\in\mathbb N$.
 
The qualitative shape of this sawtooth curve, however, does not by itself determine how the Monte Carlo test compares to the exact permutation test. Depending on $(n,p_1)$, the Monte Carlo power can lie above or below the exact power curve for substantial ranges of $B$. We illustrate these two contrasting behaviors below and also in \Cref{sec: convergence}.

\paragraph{Case 1: Exact test has larger power ($n=15$).}
Take $n=15$ and $p_1=0.16$ at level $\alpha=0.05$. In this setting,
\begin{align*}
q(3) \approx 0.1121 > \alpha,
\qquad
q(4) \approx 0.0498 < \alpha,
\end{align*}
so the exact permutation test rejects if and only if $S\ge4$. Consequently,
\begin{align*}
\mathrm{Pow}_{\mathrm{exact}} = \mP(S\ge4) \approx 0.209218.
\end{align*}
In this configuration, the Monte Carlo power curve lies entirely below the exact power curve for all $B$ (see the left panel of \Cref{fig:bernoulli_two_cases}), although it converges to $\mathrm{Pow}_{\mathrm{exact}}$ as $B\to\infty$, as shown in \Cref{sec: convergence} for a broader range of $B$.
 
\paragraph{Case 2: Exact test has smaller power ($n=25$).}
Now let $n=25$ and $p_1=0.16$ at the same level $\alpha=0.05$. Then $S\sim\mathrm{Binomial}(25,0.16)$ and
\begin{align*}
q(4) \approx 0.0549 > \alpha,
\qquad
q(5) \approx 0.0251 < \alpha.
\end{align*}
Hence the exact permutation test rejects if and only if $S\ge5$, giving
\begin{align*}
\mathrm{Pow}_{\mathrm{exact}} = \mP(S\ge5) \approx 0.370668.
\end{align*}
In this configuration, for moderate values of $B$, the Monte Carlo test can be strictly more powerful than the exact test (see the right panel of \Cref{fig:bernoulli_two_cases}). Nevertheless, as $B\to\infty$, the Monte Carlo power converges to the exact power, i.e., $\mathrm{Pow}(B) \to \mathrm{Pow}_{\mathrm{exact}}$ (see \Cref{sec: convergence} for extended results).
 
Appendix~\ref{sec: numerical_examples} further demonstrates that this non-monotonic behavior also arises for a wider class of test statistics commonly used in practical applications.

\section{Main takeaways} \label{sec: discussion}
Our results should not be interpreted as ``more permutations are bad.'' 
In many respects, the opposite is true: increasing the Monte Carlo budget $B$ reduces the simulation variability of the Monte Carlo $p$-value and makes $p_B(\calX)$ a more stable approximation to the full-group permutation $p$-value. It also refines the discrete set of attainable Monte Carlo $p$-values. The focus of this paper is more specific. Because the rejection threshold is discrete, increasing $B$ does not in general guarantee a monotone increase in power. With that in mind, the main practical messages are as follows.

\begin{itemize}
\item The number of permutations $B$ is part of the test definition and should be prespecified. Although Monte Carlo error decreases as $B$ grows, the induced decision rule $\1\{p_B(\calX)\le \alpha\}$ need not become more powerful in a monotone way. Thus ``more'' does not necessarily mean ``better'' when power at level $\alpha$ is the primary objective. A natural design recommendation is to choose $B$, subject to computational constraints, so that $(B+1)\alpha\in\mathbb{N}$. This alignment is consistent with earlier analyses \citep{hope1968simplified,jockel1986finite}, and our results provide a complementary structural reason for it. 
\item If one is willing to introduce auxiliary randomization, we recommend using the randomized Monte Carlo permutation $p$-value
\begin{align*}
  p_B^{\mathrm{rand}}(\calX) = p_B(\calX) - U\frac{1 + \sum_{i=1}^B \1\{T(\calX^{\pi_i}) = T(\calX)\}}{B+1},
\end{align*}
where $U\sim\mathrm{Uniform}(0,1)$ is independent of $\calX$. Under the null, this test rejects with probability exactly $\alpha$ \citep[Proposition~3]{hemerik2018exact} and its power is no smaller than that of the nonrandomized test. Intuitively, the extra randomization smooths the discreteness and correspondingly dampens the sawtooth pattern in the power curve.
\item Rather than fixing $B$ in advance, sequential procedures \citep{besag1991sequential,gandy2009sequential,fischer2025sequential} can adaptively determine how many permutations to draw, based on the observed data, while controlling type I error at level $\alpha$. These methods can be more computationally efficient than fixed-$B$ tests when the decision is relatively clear, that is, when the resulting $p$-value is likely to be very small or very large.

\end{itemize}
 
We hope that our results will encourage practitioners to think more carefully about the choice of $B$ and to consider alternative approaches when appropriate. In particular, the non-monotonicity phenomenon highlights the importance of understanding the interplay between computational resources and statistical properties in permutation testing.

\bibliographystyle{apalike} 
\bibliography{reference}
 
\appendix

\newpage
\section{Proofs}
This section contains the proofs of \Cref{thm:nonmonotonicity} and \Cref{prop:decrease_size_uncond}.

\subsection{Proof of \Cref{thm:nonmonotonicity}}
\begin{proof}
Recall that $k_B = \lfloor(B+1)\alpha\rfloor - 1$. For $B \geq 2$, the difference between successive critical counts is given by
\begin{align*}
k_B - k_{B-1} =
\begin{cases}
1, & \text{if } \lfloor (B+1)\alpha\rfloor = \lfloor B\alpha\rfloor + 1,\\
0, & \text{otherwise}.
\end{cases}
\end{align*}
We prove the two claims in turn. Fix $B \geq 2$ such that
\begin{align*}
k_B-k_{B-1}=1
\qquad\text{and}\qquad
k_{B+1}-k_B=0.    
\end{align*}
Let $A:=\{q(\calX)<1\}$, which satisfies $\mP(A)>0$ by \Cref{assumption: strict prob}. Conditionally on $\calX$, let
\begin{align*}
R_{B-1}\sim\mathrm{Binomial}(B-1,q(\calX)),
\qquad
R_B\sim\mathrm{Binomial}(B,q(\calX)).    
\end{align*}
Then
\begin{align*}
\phi_{B-1}(\calX)
=\mP(R_{B-1}\le k_{B-1}\mid\calX),
\qquad
\phi_B(\calX)
=\mP(R_B\le k_B\mid\calX).    
\end{align*}
Since $k_B=k_{B-1}+1$, moving from $B-1$ to $B$ adds one Bernoulli trial and relaxes the rejection threshold by one. Couple $R_B = R_{B-1} + Y$, where $Y \sim \mathrm{Bernoulli}(q(\calX))$ is independent of $R_{B-1}$ given $\calX$. Then we can decompose the probability as follows:
\begin{align*}
\phi_B(\calX)
&= \mP(R_{B-1} + Y \le k_{B-1} + 1 \mid \calX) \\
&= \mP(Y=0 \mid \calX) \mP(R_{B-1} \le k_{B-1} + 1 \mid \calX) + \mP(Y=1 \mid \calX) \mP(R_{B-1} \le k_{B-1} \mid \calX) \\
&= (1-q(\calX)) \big( \mP(R_{B-1} \le k_{B-1} \mid \calX) + \mP(R_{B-1} = k_{B-1} + 1 \mid \calX) \big) \\
&\quad + q(\calX) \mP(R_{B-1} \le k_{B-1} \mid \calX) \\
&= \mP(R_{B-1} \leq k_{B-1} \mid \calX) + (1-q(\calX))\mP(R_{B-1} = k_{B-1} + 1 \mid \calX) \\
&= \phi_{B-1}(\calX) + (1-q(\calX))\mP(R_{B-1} = k_{B-1} + 1 \mid \calX).
\end{align*}
On the set $A$, we have $1-q(\calX)>0$. Furthermore, since $k_{B-1} + 1 = k_B \leq B-1$, the binomial mass $\mP(R_{B-1} = k_{B-1} + 1 \mid \calX)$ is strictly positive. This strictly increases the binomial left-tail probability:
\begin{align*}
\phi_{B-1}(\calX)<\phi_B(\calX)
\quad\text{on }A,    
\end{align*}
while equality holds on $A^c$ where $q(\calX)=1$. 
Taking expectations gives
\begin{align*}
\mathrm{Pow}(B-1)<\mathrm{Pow}(B).    
\end{align*}
Next, since $k_{B+1}=k_B=:k$, we are on a plateau. Couple
\begin{align*}
R_{B+1}=R_B+Y',
\qquad
Y'\sim\mathrm{Bernoulli}(q(\calX)),    
\end{align*}
with $Y'$ independent of $R_B$ given $\calX$. Then
\begin{align*}
\phi_{B+1}(\calX)
&= \mP(R_B+Y'\le k\mid\calX) \\
&= \mP(R_B \le k \mid \calX) - \mP(Y'=1 \mid \calX) \mP(R_B = k \mid \calX) \\
&= \phi_B(\calX) - q(\calX)\,\mP(R_B=k\mid\calX).
\end{align*}
For $q(\calX)\in(0,1)$ and $0\le k\le B$ (note that $k \ge 0$ is naturally guaranteed by the preceding jump $k_B = k_{B-1} + 1 \ge 0$), the binomial mass $\mP(R_B=k\mid\calX)$ is strictly positive, hence
\begin{align*}
\phi_{B+1}(\calX)<\phi_B(\calX)
\quad\text{on }A,    
\end{align*}
with equality again on $A^c$. Taking expectations yields
\begin{align*}
\mathrm{Pow}(B+1)<\mathrm{Pow}(B).    
\end{align*}
Thus $\mathrm{Pow}(B)$ is a strict local maximum, proving part~(i).

\medskip
\noindent For part~(ii), define the backward difference
\begin{align*}
d_B \coloneqq k_B - k_{B-1} \in \{0,1\} \qquad \text{for } B \geq 2.
\end{align*}
Hence $(k_B)$ is nondecreasing and increases by at most one at each step. Moreover,
\begin{align*}
(B+1)\alpha-2<k_B \le (B+1)\alpha-1,    
\end{align*}
so that $k_B/B\to\alpha\in(0,1)$.
 
We first show that both values of $d_B$ occur infinitely often. If $d_B=1$ for all sufficiently large $B$, then for some $B_0$ and all $B>B_0$ we would have
\begin{align*}
k_B
=
k_{B_0}
+
\sum_{j=B_0+1}^{B} d_j
=
k_{B_0}+(B-B_0),
\end{align*}
so that
\begin{align*}
\frac{k_B}{B}
=
1+\frac{k_{B_0}-B_0}{B}
\rightarrow 1,
\end{align*}
contradicting $k_B/B\to\alpha<1$. Thus $d_B=0$ must occur infinitely often. A symmetric argument shows that $d_B=1$ must also occur infinitely often, since if $d_B=0$ eventually, then $k_B$ would be eventually constant and hence $k_B/B\to0$, contradicting $\alpha>0$.
 
We now show that the pattern $(d_{B},d_{B+1})=(1,0)$ occurs infinitely often. Suppose, for contradiction, that it occurs only finitely many times. Then there exists $B_1$ such that for all $B\ge B_1$,
\begin{align*}
d_{B}=1 \;\implies\; d_{B+1}=1.
\end{align*}
In other words, beyond $B_1$, once a jump occurs it must be followed by another jump.
 
Since $d_B=1$ occurs infinitely often, choose $B^\star\ge B_1$ such that $d_{B^\star}=1$. 
By the implication above, we then have
\begin{align*}
d_{B^\star+1}=1,\quad
d_{B^\star+2}=1,\quad
\ldots
\end{align*}
and thus $d_B=1$ for all $B\geq B^\star$. 
But this contradicts the fact that $d_B=0$ occurs infinitely often. Therefore the implication cannot hold, and the pattern $(1,0)$ must occur infinitely many times.
 
By part~(i), each such index $B$ yields a strict local maximum of $\mathrm{Pow}(B)$. Hence the power curve has infinitely many strict local maxima.
\end{proof}

\subsection{Proof of \Cref{prop:decrease_size_uncond}}
\begin{proof} 
Write $q\coloneqq q(\calX)$ and assume $k_{B+1}=k_B=:k$. If $k=-1$, then $\phi_B(\calX)=\mP(R_B\le -1\mid\calX)=0$ and likewise $\phi_{B+1}(\calX)=0$, meaning the power difference is trivially zero and the bound holds. Thus, we may assume $k\ge 0$. Since $\alpha \in (0,1)$, we also naturally have $k = \lfloor(B+1)\alpha\rfloor-1 \le B-1$. 
 
Conditional on $\calX$, we couple the binomial variables as $R_{B+1}=R_B+Y$, where $R_B\mid\calX\sim\mathrm{Binomial}(B,q)$ and $Y\mid\calX\sim\mathrm{Bernoulli}(q)$ are independent. Following the exact same logic as in the proof of \Cref{thm:nonmonotonicity}, the difference in conditional power is exactly $q\,\mP(R_B=k\mid\calX)$. Taking expectations over $\calX$ yields the claimed identity:
\begin{align*}
  \mathrm{Pow}(B)-\mathrm{Pow}(B+1) = \mE\big[q(\calX)\,\mP(R_B=k \mid \calX)\big].    
\end{align*}
To bound this quantity uniformly, let $b(k; B, q) \coloneqq \binom{B}{k} q^k (1-q)^{B-k}$ denote the binomial probability mass function. The expected decrease is bounded by the supremum $\sup_{q\in[0,1]} q\cdot b(k; B, q)$. To find this supremum, we maximize the logarithm of the objective, $\log\bigl(q^{k+1}(1-q)^{B-k}\bigr) = (k+1)\log q + (B-k)\log(1-q)$. Taking the derivative with respect to $q$ and setting it to zero gives
\begin{align*}
  \frac{k+1}{q}-\frac{B-k}{1-q} = 0 \implies (k+1)(1-q) = (B-k)q \implies k+1 = (B+1)q.     
\end{align*}
This vanishes uniquely at $q^\star=(k+1)/(B+1)$. Since the log function is strictly concave on $(0,1)$, this critical point is the global maximizer. 
 
Setting $m\coloneqq k+1$ and $N\coloneqq B+1$, we have $q^\star = m/N$. We can relate the binomial mass at $B$ to the mass at $N$ through the identity $\binom{B}{k} = \binom{B}{m-1} = \frac{m}{B+1}\binom{B+1}{m} = q^\star \binom{N}{m}$. Using this relation, we observe the following crucial simplification:
\begin{align*}
  b(k; B, q^\star)
  &= \binom{B}{k} (q^\star)^k (1-q^\star)^{B-k} \\
  &= \bigg\{ q^\star \binom{N}{m} \bigg\} (q^\star)^{m-1} (1-q^\star)^{N-m} \\
  &= \binom{N}{m} (q^\star)^m (1-q^\star)^{N-m} = b(m; N, q^\star).
\end{align*}
 
We bound this probability mass using Robbins' bounds for Stirling's formula~\citep{robbins1955remark}, which guarantee that $n!$ is bounded between $\sqrt{2\pi}\,n^{n+1/2}e^{-n+1/(12n+1)}$ and $\sqrt{2\pi}\,n^{n+1/2}e^{-n+1/(12n)}$ for any integer $n\ge 1$. Since $0 \le k \le B-1$, both $m \ge 1$ and $N-m \ge 1$ are satisfied, allowing us to apply these bounds to $N!$, $m!$, and $(N-m)!$. The terms $e^{-N}$ in the numerator and $e^{-m}e^{-(N-m)}$ in the denominator cancel out perfectly. Substituting these into the definition of $b(m; N, q^\star)$ and extracting the fractional powers yields
\begin{align*}
  b(m; N, q^\star)
  &\le
  \frac{\sqrt{2\pi}\, N^{N+1/2}}
       {\sqrt{2\pi}\, m^{m+1/2}\, \sqrt{2\pi}\, (N-m)^{N-m+1/2}}
  \exp\Bigl(
    \tfrac{1}{12N}
    - \tfrac{1}{12m+1}
    - \tfrac{1}{12(N-m)+1}
  \Bigr) \\
  &\qquad \times (q^\star)^m (1-q^\star)^{N-m}.
\end{align*}
By recognizing that $(q^\star)^m (1-q^\star)^{N-m} = (m/N)^m ((N-m)/N)^{N-m} = m^m (N-m)^{N-m} / N^N$, we can explicitly cancel all the terms with exponents of $m$, $N-m$, and $N$. After cancellation, the bound simplifies to
\begin{align*}
  b(m; N, q^\star) \le \frac{1}{\sqrt{2\pi}}\sqrt{\frac{N}{m(N-m)}} \exp\biggl(\frac{1}{12N}-\frac{1}{12m+1}-\frac{1}{12(N-m)+1}\biggr).    
\end{align*}
The argument inside the exponential function is strictly negative because
\begin{align*}
  \frac{1}{12m+1}+\frac{1}{12(N-m)+1} \ge \frac{4}{(12m+1)+(12(N-m)+1)} = \frac{2}{6N+1} > \frac{1}{12N}.    
\end{align*}
Consequently, the exponential factor is strictly bounded by $1$. Since $\frac{m}{N} = q^\star$ and $\frac{N-m}{N} = 1-q^\star$, we can rewrite the remaining square root as $(2\pi N\,q^\star(1-q^\star))^{-1/2}$, meaning $b(k; B, q^\star) \le (2\pi N\,q^\star(1-q^\star))^{-1/2}$. Multiplying both sides by $q^\star$ gives
\begin{align*}
  q^\star\,b(k; B, q^\star) \le \frac{1}{\sqrt{2\pi (B+1)}}\sqrt{\frac{q^\star}{1-q^\star}}.    
\end{align*}
Finally, under our plateau assumption $k=\lfloor(B+1)\alpha\rfloor-1$, we have $q^\star \le \alpha$. Because the function $x\mapsto x/(1-x)$ is strictly increasing on $(0,1)$, we may replace $q^\star$ by $\alpha$ to achieve the desired distribution-free bound
\begin{align*}
  \mathrm{Pow}(B)-\mathrm{Pow}(B+1) \leq \frac{1}{\sqrt{2\pi (B+1)}}\sqrt{\frac{\alpha}{1-\alpha}}.    
\end{align*}
This proves the claimed bound.
\end{proof}

\section{Numerical examples}\label{sec: numerical_examples}

\begin{figure}[t!]
\centering
\includegraphics[width=0.9\textwidth]{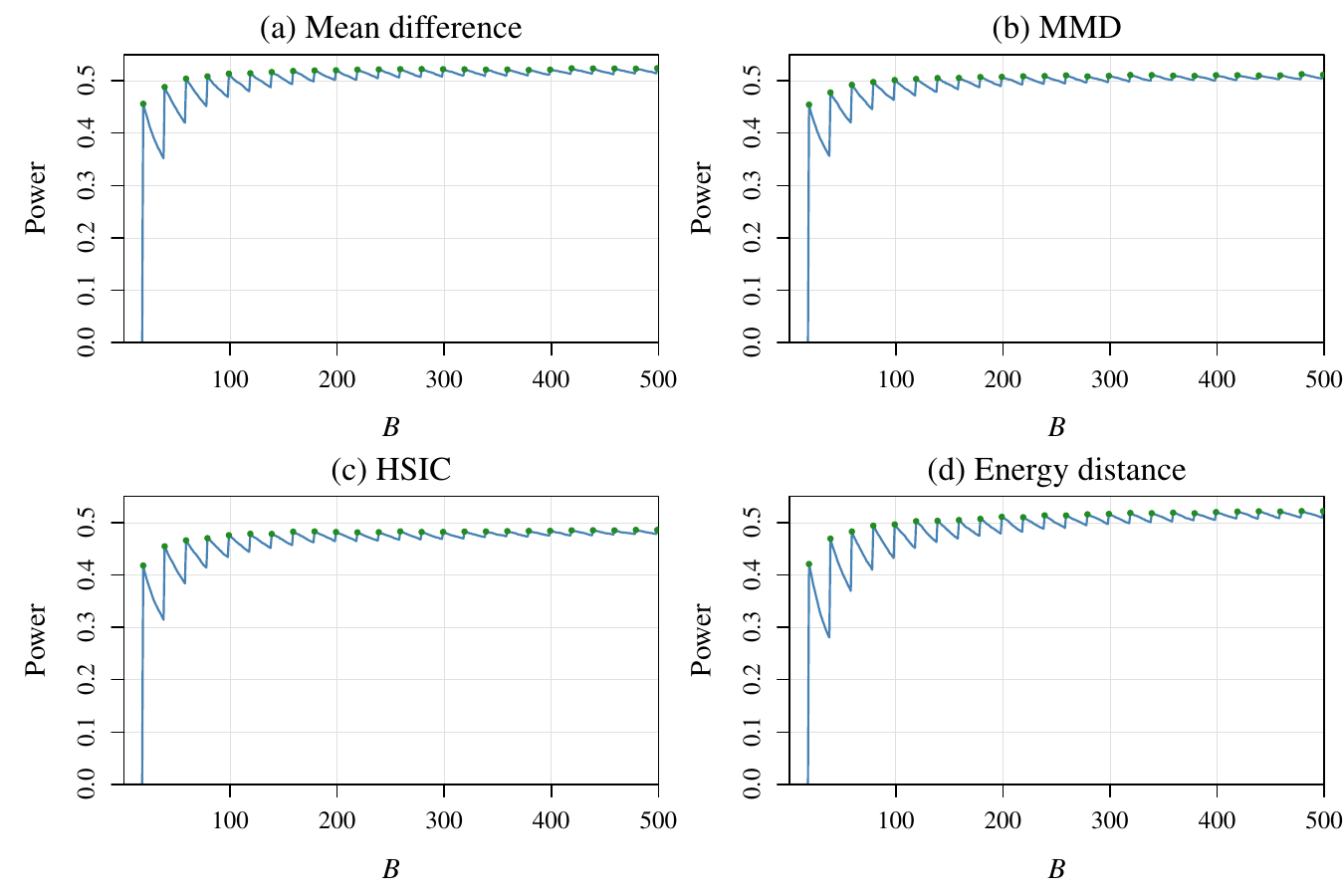}
\caption{Estimated power curves for four permutation test statistics at level $\alpha = 0.05$, based on $N_{\mathrm{sim}} = 10{,}000$ replications per value of $B$. Dots mark the predicted local maxima where $(B+1)\alpha \in \mathbb{N}$.}
\label{fig:numerical_examples}
\end{figure}

The closed-form Bernoulli example in \Cref{sec: concrete example} isolates the sawtooth mechanism in a setting where the unconditional power admits an exact finite-sum representation. In this appendix, we demonstrate that the same non-monotonic behavior arises across a range of commonly used permutation test statistics, covering both two-sample and independence testing problems.

For each test statistic and each $B\in\{1,\ldots,500\}$, we estimate $\mathrm{Pow}(B)$ by drawing $10{,}000$ datasets from the alternative and running the Monte Carlo permutation test at level $\alpha=0.05$.

\vspace{0.4em}
\noindent\textbf{(a) Mean difference. }
A two-sample location problem with $n_1=n_2=20$ per group. Under the alternative, the two samples are drawn from $N(\delta,1)$ and $N(0,1)$ with $\delta=0.55$. The test statistic is the difference in sample means.

\vspace{0.2em}
\noindent\textbf{(b) MMD. }
A kernel two-sample test based on the maximum mean discrepancy (MMD) \citep{gretton2012kernel} with $n_1=n_2=25$ and dimension $d=5$. Under the alternative, the two samples are drawn from $N(\delta\mathbf{1}_d, I_d)$ and $N(\mathbf{0}_d, I_d)$ with $\delta=0.35$, where $\mathbf{1}_d$ and $\mathbf{0}_d$ denote the $d$-dimensional vectors of ones and zeros. The test statistic is the unbiased squared MMD estimator with the Gaussian kernel, with bandwidth set to the median of pairwise distances.

\vspace{0.2em}
\noindent\textbf{(c) HSIC. }
An independence test based on the Hilbert--Schmidt independence criterion (HSIC) \citep{gretton2005measuring} with $n=50$ paired observations. Under the alternative, $X_i\iid N(0,1)$ and $Y_i=X_i^2+\varepsilon_i$ with $\varepsilon_i\iid N(0,\sigma_\varepsilon^2)$ and $\sigma_\varepsilon=2.5$. The test statistic is the empirical HSIC with Gaussian kernels on both marginals. Permutations act by independently reshuffling the $Y$ coordinates while keeping the $X$ coordinates fixed.

\vspace{0.2em}
\noindent\textbf{(d) Energy distance. }
A two-sample test based on the energy distance \citep{szekely2004testing} with $n_1=n_2=20$. Under the null, both samples are drawn from $N(0,1)$. Under the alternative, the two samples are drawn from $N(0,1)$ and $N(0,\sigma_{\mathrm{alt}}^2)$ with $\sigma_{\mathrm{alt}}=1.8$, representing a pure scale shift. The test statistic is the energy distance based on Euclidean distances.

\vspace{0.4em}
\noindent\textbf{Results. }
All four panels of \Cref{fig:numerical_examples} exhibit the sawtooth pattern predicted by \Cref{thm:nonmonotonicity}. In each case, the power increases at values of $B$ where the integer cutoff $k_B$ jumps, and then decreases along the following plateau where the cutoff remains fixed. The observed local maxima occur exactly at the predicted values satisfying
$(B+1)\alpha \in \mathbb{N}$, which are marked by dots in the figure. This agreement supports the view that the oscillation is driven by the discreteness of the Monte Carlo rejection threshold, rather than by model-specific features of the data generating distribution.

The size of the oscillation is most pronounced for small $B$ and becomes progressively smaller as $B$ increases. This behavior is consistent with the $O(B^{-1/2})$ upper bound in \Cref{prop:decrease_size_uncond}. Taken together, these examples indicate that the phenomenon is a structural consequence of using a finite Monte Carlo permutation distribution with a discrete rejection threshold, rather than a peculiarity of any single test statistic or alternative distribution.

\section{Convergence of the Monte Carlo power to the exact-test power}
\label{sec: convergence}
\begin{figure}[t!]
  \centering
  \includegraphics[width=0.9\textwidth]{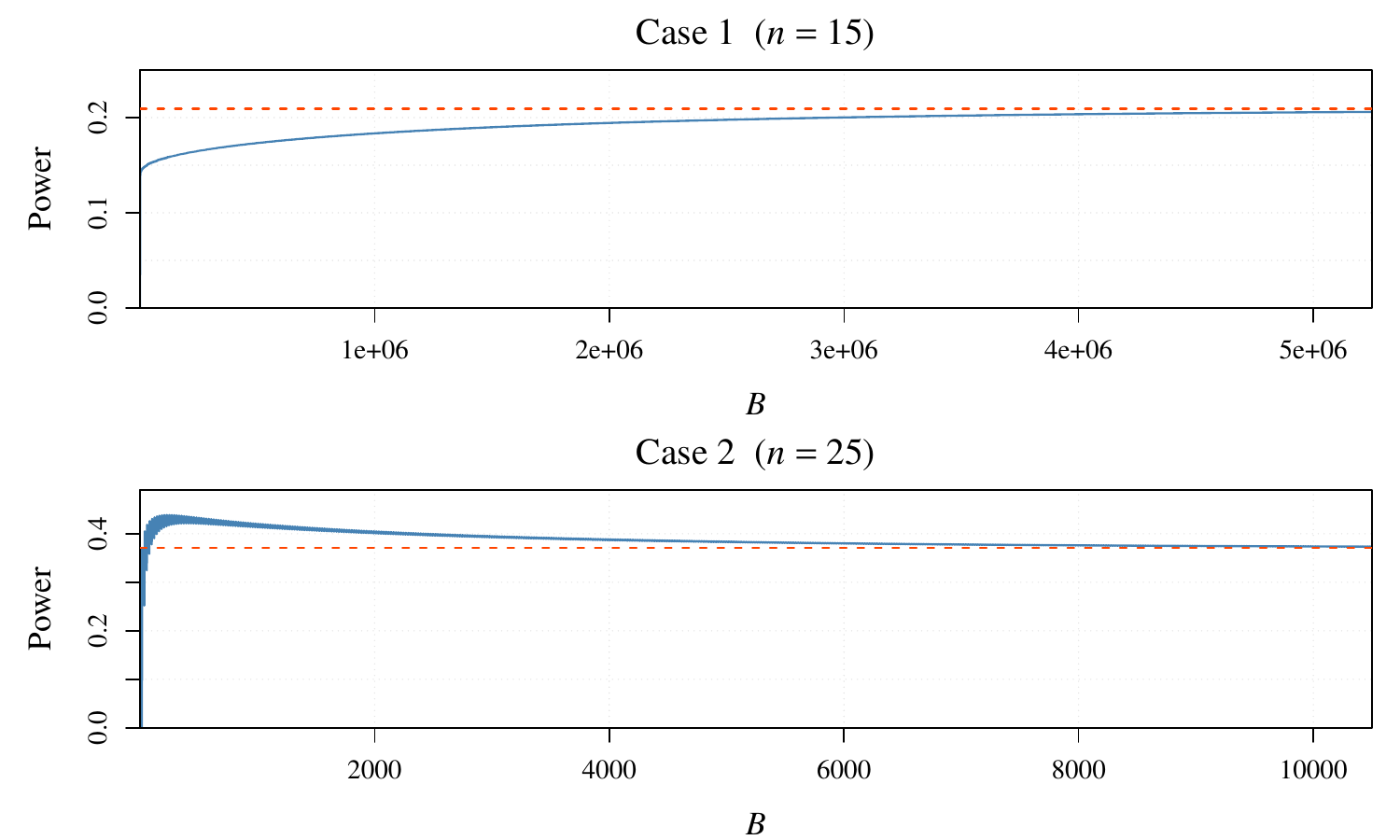}
  \caption{Unconditional power $\mathrm{Pow}(B)$ for the closed-form Bernoulli example at level $\alpha = 0.05$ with $p_0 = 0$ and $p_1 = 0.16$. The top and bottom panels show Case~1 ($n = 15$, $B \leq 5 \times 10^{6}$) and Case~2 ($n = 25$, $B \leq 10^{4}$). Solid and dashed lines represent $\mathrm{Pow}(B)$ and $\mathrm{Pow}_{\mathrm{exact}}$, respectively.}
  \label{fig:convergence}
\end{figure}

To examine how the Monte Carlo power approaches its exact-test counterpart as $B$ grows, we extend the closed-form Bernoulli example of \Cref{sec: concrete example} to $B = 5 \times 10^{6}$ for Case~1 and $B = 10^{4}$ for Case~2. \Cref{fig:convergence} shows the unconditional power over this extended range. In both cases, the Monte Carlo power converges to $\mathrm{Pow}_{\mathrm{exact}}$. The sawtooth oscillations guaranteed by \Cref{thm:nonmonotonicity} persist for all $B$, but their amplitude decays according to the $O(B^{-1/2})$ bound of \Cref{prop:decrease_size_uncond}.

The rate of convergence differs markedly between cases and depends on the proximity of the critical exceedance probability to $\alpha$. In Case~1, $q(4) \approx 0.0498$ is just below $\alpha = 0.05$, so the conditional rejection probability $\mathbb{P}(R_B \leq k_B \mid S = 4)$ is sensitive to $k_B$ and stabilizes only for $B$ around $10^6$. In Case~2, $q(5) \approx 0.0251$ is far from $\alpha$, and convergence is almost complete by $B \approx 10^4$.

\end{document}